\documentclass[journal,onecolumn]{IEEEtran}

\usepackage{cite}
\usepackage{amsmath,amssymb,amsfonts}
\usepackage{graphicx}
\usepackage{url}
\usepackage{textcomp}
\usepackage{xcolor}
\usepackage{physics}
\usepackage{mathrsfs}
\usepackage[caption=false,font=footnotesize]{subfig}
\usepackage{hyperref}
\hypersetup{hidelinks=true}
\allowdisplaybreaks

\newtheorem{thm}{Theorem}
\newtheorem{prop}{Proposition}

\newtheorem{lem}{Lemma}

\newtheorem{assum}{Assumption}
\newtheorem{rem}{Remark}
\newenvironment{pf}{\begin{IEEEproof}}{\end{IEEEproof}}

\newcommand{\sgp}{\sigma^+}
\newcommand{\sgn}{\sigma^-}
\newcommand{\delsgp}{\delta_{\sigma^+_j}}
\newcommand{\delsgn}{\delta_{\sigma^-_j}}
\newcommand{\Kuu}{K^{uu}}
\newcommand{\Kuv}{K^{uv}}
\newcommand{\Kvu}{K^{vu}}
\newcommand{\Kvv}{K^{vv}}
\newcommand{\hKuu}{\widehat{K}^{uu}}
\newcommand{\hKuv}{\widehat{K}^{uv}}
\newcommand{\hKvu}{\widehat{K}^{vu}}
\newcommand{\hKvv}{\widehat{K}^{vv}}
\newcommand{\tKuu}{\widetilde{K}^{uu}}
\newcommand{\tKuv}{\widetilde{K}^{uv}}
\newcommand{\tKvu}{\widetilde{K}^{vu}}
\newcommand{\tKvv}{\widetilde{K}^{vv}}
\newcommand{\Laa}{L^{\alpha\alpha}}
\newcommand{\Lab}{L^{\alpha\beta}}
\newcommand{\Lba}{L^{\beta\alpha}}
\newcommand{\Lbb}{L^{\beta\beta}}

\newcommand{\heone}{\widehat e_{j,1}}
\newcommand{\hetwo}{\widehat e_{j,2}}
\newcommand{\rone}{r^\epsilon_{j,1}}
\newcommand{\rtwo}{r^\epsilon_{j,2}}
\newcommand{\inner}[2]{\left\langle #1, #2 \right\rangle}
\newcommand{\e}{\mathrm{e}}

\def\BibTeX{{\rm B\kern-.05em{\sc i\kern-.025em b}\kern-.08em
    T\kern-.1667em\lower.7ex\hbox{E}\kern-.125emX}}

\begin{document}
\title{Event-Triggered Gain Scheduling of $2\times 2$ Linear  Hyperbolic PDEs via Neural Operators~(Full Version)}
\author{Yihuai Zhang, Jean Auriol, Nicolas Espitia, and Huan Yu
\thanks{This work was supported by the National Science Foundation of China under Grant 12526214 and 62203131. The work of the second and third authors was partially funded by the Agence Nationale de la Recherche (ANR) via grants PANOPLY ANR-23-CE48-0001-01 and  PH-DIPSY ANR-24-CE48-1712. Corresponding author: Huan Yu.}
\thanks{Yihuai Zhang and Huan Yu are with the Hong Kong University of Science and Technology (Guangzhou), Thrust of Intelligent Transportation, Guangzhou, China (e-mail: yzhang169@connect.hkust-gz.edu.cn; huanyu@hkust-gz.edu.cn).}
\thanks{Jean Auriol is with Universit\'{e} Paris-Saclay, CNRS, CentraleSup\'{e}lec, Laboratoire des Signaux et Syst\`{e}mes, 91190, Gif-sur-Yvette, France (e-mail: jean.auriol@centralesupelec.fr).}
\thanks{Nicolas Espitia is with Univ. Lille, CNRS, Centrale Lille, UMR 9189 CRIStAL, F-59000 Lille, France (e-mail: nicolas.espitia-hoyos@univ-lille.fr).}}

\maketitle

\begin{abstract}
This paper introduces a new framework for event-triggered gain scheduling applied to linear hyperbolic Partial Differential Equations (PDEs) with time- and space-varying coefficients. The approach leverages neural operators to address the challenges of real-time control in such systems. At each triggering time, the control input is designed using the classical static backstepping control law, while the gains of the boundary controller are updated according to the triggering mechanism and the spatial variation of the coefficients. Neural operators are employed to learn the mapping between the system parameters in the PDEs and the corresponding backstepping kernels. By integrating neural operators into the event-triggered framework, we eliminate the need to repeatedly solve complex kernel equations at every triggering instant, thereby reducing computational overhead while ensuring closed-loop stability. The proposed method is validated through theoretical analysis and numerical simulations,  demonstrating its effectiveness and strong potential for real-time control of time-varying hyperbolic PDE systems.
\end{abstract}


\section{Introduction}
Boundary control of hyperbolic Partial Differential Equations (PDEs) is widely applied to engineering  problems that require point actuation for spatial-temporal stabilization, such as oil drilling~\cite{wang2020delay}, traffic flow~\cite{yu2022traffic}, gas pipes~\cite{bastin2016stability}. PDE backstepping achieves Lyapunov stabilization by Volterra spatial transformation and then eliminates destabilizing in-domain terms by boundary feedback controller design~\cite{krstic2008boundary}.  {For $2\times2$ hyperbolic systems, this methodology is supported by constructive kernel equations and explicit boundary feedback designs~\cite{vazquez_backstepping_2011}; see also the recent survey~\cite{vazquez2026backstepping}.} A major challenge arises when the PDE coefficients depend on time and space. In such cases, the backstepping kernels and the corresponding controller gains inherit the time dependence of the coefficients and, in principle, must be updated online to preserve stability. This requirement creates a significant implementation burden: recomputing kernels continuously entails solving coupled kernel PDEs repeatedly which leads to increasing computational complexity. Compared with the extensive literature on time invariant hyperbolic PDEs, fewer results are available for backstepping control of time-varying hyperbolic PDEs, with existing studies such as observer design for hyperbolic Partial Integro-Differential Equation~\cite{deutschmann2016backstepping}, stabilization and tracking~\cite{anfinsen2020stabilization},as well as finite-time boundary stabilization for general hyperbolic systems of balance laws with both time- and space-varying coefficients \cite{CoronBoundaryStabilization2021}.  {Related work also includes infinite-dimensional backstepping control for time-variant one-dimensional hyperbolic PDEs~\cite{anfinsen2018control}. These contributions demonstrate that time-varying coefficients can be handled rigorously, while also highlighting the additional analytical and computational complexity of time-dependent kernel equations.}

However, the kernel equations for the time-varying case are more complex than the time-invariant case. They may not admit closed-form solutions. Hence, a key challenge for real-time implementation is the need to solve time-varying kernel equations online. This problem was addressed in \cite{espitia_event-triggered_2025} which introduced an event-triggered gain scheduling (ETGS) strategy for $2\times2$ linear hyperbolic PDEs with time and space varying in-domain coupling coefficients. The key idea is to freeze the time-varying coefficients at triggering instants and, between events, update the kernels and apply a static (space-dependent) backstepping controller computed for the frozen coefficients. That approach drew inspiration by the framework of Event-triggered-Control (ETC), namely event-triggered  PDE backstepping method \cite{espitia_event-based_2018}  (see some recent contributions on ETC e.g. \cite{wang2021event,Somathilake2025,Rathnayake2026} and references therein, in the context of \textit{hyperbolic} PDEs).  {More broadly, event-based control for hyperbolic conservation laws was studied in~\cite{espitia2016event}, and related event-triggered designs have been developed for parabolic PDEs, PDE--ODE cascades, Stefan problems, and traffic-flow models~\cite{wang2022event,rathnayake2024observer,rathnayake2024observer2,rathnayake2025performance,zhang2025performance,zhang2025event,espitia2022traffic}. These works show that triggering mechanisms can preserve stability while reducing unnecessary updates, but most of them treat the controller gains as fixed after an offline design or updated through a prescribed sampled-data structure.}

Despite this important step, a critical computational bottleneck remains: at each triggering instant, the backstepping kernels for the frozen coefficients must still be computed. In~\cite{espitia_event-triggered_2025}, this is done by numerically solving e.g., via successive approximations the kernel equations at every event. Although events occur less frequently than continuous updates, kernel computation may still dominate the online cost that degrade performance when events become frequent (i.e., when the minimum inter-event time is small) or when the plant dynamics are fast. Hence, a key missing component for real-time deployment is a mechanism to generate accurate backstepping kernels almost instantaneously at triggering instants.  {In this sense, ETGS connects the resource-aware nature of ETC with the need to update stabilizing gains when plant parameters vary. For reaction-diffusion PDEs, event-triggered gain scheduling was developed to update controller gains only when the scheduling error becomes significant~\cite{KarafyllisEventtriggeredGain2021}.}

This paper addresses the above gap by integrating neural operators (NOs) into the event-triggered gain scheduling loop. Neural operators can approximate nonlinear operator mappings between function spaces.  {Their approximation-theoretic foundations are connected to universal approximation results for nonlinear operators~\cite{chen1995universal,lu2021learning}, with further analysis for DeepONet-type operator learning in PDE settings~\cite{deng2022approximation}.} It has begun to be used to accelerate gain scheduling~\cite{lamarque_gain_2025} and backstepping-related computations~\cite{bhan2023neural,wang_backstepping_2025} in PDE control.  {In PDE control, neural operators have also been used to approximate controller and observer gain functions for reaction-diffusion PDEs~\cite{KrsticNeuralOperators2024} and to handle first-order hyperbolic PIDE backstepping with recycle and delay~\cite{qi_neural_2024}, as well as hyperbolic PDEs with Markov-jumping parameters~\cite{zhang_operator_2025}. These developments indicate that operator learning can replace repeated online solution of kernel equations by an offline-trained surrogate, provided that the resulting approximation error is explicitly accounted for in the closed-loop analysis.}

However, neural operators for event-triggered gain scheduling of PDEs remain unexplored, motivating their use as fast surrogate solvers for the kernel computation.  {To the best of our knowledge, neural operators have not yet been combined with ETGS for this class of systems. The objective is therefore to preserve the stability and Zeno-free features of event-triggered gain scheduling while replacing the online kernel solver by a fast neural-operator surrogate and explicitly incorporating the approximation error into the Lyapunov analysis.}

\textbf{Contributions.}
We propose a NO-accelerated event-triggered gain scheduling framework for stabilizing $2\times2$ linear hyperbolic PDEs with time- and space-varying in-domain coupling coefficients. By training a neural operator to learn the mapping from the system parameters to the backstepping kernels, approximated kernels are obtained at triggering instants through trained NO with low online computational cost. The main contribution lies in the use of NO for fast kernel calculation at triggering events. We introduce a Lyapunov-based triggering condition design and prove the stability of the system with NO. It paves the way for real-time implementation of backstepping-based controllers for time-varying hyperbolic PDEs.

The paper is organized as follows: Section~\ref{sec:Pro} introduces the problem formulation and control objective. Section~\ref{sec:Preliminaries} gives the preliminaries of control design for hyperbolic system with time-independent coefficients, as well as preliminaries on neural operators. Section~\ref{sec:timevary} states  the problem in the time-varying coefficients system and introduces the gain-scheduling mechanism. Section~\ref{sec:ETGS} presents the NO-accelerated event-triggered gain scheduling design and the stability proofs, as well as Zeno-free guarantees. Section~\ref{sec:Sim} reports numerical simulations, and Section~\ref{sec:Con} concludes the paper and outlines future directions.

\textbf{Notation.}
For a scalar, vector, or matrix-valued function $f(\cdot)$ defined on a domain $\Omega$, $\|f\|$ denotes the $L^2$-norm, and $\|f\|_\infty$ is used interchangeably with $\|f\|_{L^{\infty}}$ when the domain is clear from the context.

\section{Problem formulation}\label{sec:Pro}
We consider a  $2 \times 2$ linear hyperbolic system described by the following PDEs
\begin{align}
 \partial_t u(t, x)+\lambda(x) \partial_x u(t, x)&=\sigma^{+}(t, x) v(t, x),\label{eq:sys1}\\
 \partial_t v(t, x)-\mu(x) \partial_x v(t, x)&=\sigma^{-}(t, x) u(t, x),
\end{align}
with boundary conditions
\begin{align}
 u(t, 0)&=q v(t, 0), \label{eq:sys3} \\
 v(t, 1)&=\rho u(t, 1)+U(t). \label{eq:sys4} 
\end{align}
where the spatial and time variables $(t,x)$ belong to $\{t \geq t_0, x \in [0,1] \}$, where $t_0 \geq 0$ is the initial time. The function $U(t)$ is the boundary control input to be designed. The velocities $\lambda(x)>0$ and $\mu(x)>0$ belong to $\mathcal{C}^1([0,1])$ with upper and lower bound, $\underline{\lambda} \leq \lambda(x) \leq \overline{\lambda}$, $\underline{\mu}\leq\mu(x)\leq \overline{\mu}$. {The boundary coupling $\rho$ and $q$ are constants and verify $\abs{\rho q} < 1$ to guarantee the delay-robustness of the closed-loop system~\cite{AuriolRobustOutput2020}}. We assume $q\ne 0$ to simplify the backstepping transformation we use later in this paper.  The time-varying coupling terms $\sigma^+$ and $\sigma^-$ belong to $\mathcal{C}^1([t_0,\infty)\times [0,1])$ $\cap$ $L^\infty([t_0,\infty)\times [0,1])$ and are assumed to be uniformly bounded, i.e., for all $t\geq t_0$ and $x\in [0, 1]$
\begin{align}\label{eq:assumption_bound_coupling_coeff}
    \abs{\sgp(t,x)}, \abs{\sgn(t,x)} < M_{\sigma},
\end{align}
where $M_\sigma$ is a positive constant. We make the following additional assumption.
\begin{assum}\label{assum:sig}
    There exist a constant $l_\sigma > 0$ such that for all $x \in [0, 1]$, and all $t,s \geq t_0$
    \begin{align}
        \abs{\sgp(t,x)- \sgp(s,x)}&\leq  l_\sigma \abs{t-s},\\
        \abs{\sgn(t,x)- \sgn(s,x)} &\leq l_\sigma \abs{t-s}.
    \end{align}
\end{assum}
This assumption means that the different time-varying coefficients are Lipschitz with respect to time.  The well-posedness of the closed-loop system will be assessed once the control law $U(t)$ is properly designed in Section~\ref{sec:timevary}.

The controller design  in~\cite{CoronBoundaryStabilization2021}  directly handles the time-varying coupling coefficients by applying suitable time-varying backstepping transformations, resulting in much more complex kernel equations to solve and high computational burden. In this paper, our objective is to stabilize the system~\eqref{eq:sys1}-\eqref{eq:sys4} using a suitable backstepping boundary controller while avoiding solving time-varying kernel equations by combining gain scheduling schemes and neural operators.
For a given sequence of time instants $\{t_j\}_{j\in\mathbb{N}}$, which will be characterized according to the event-triggered gain scheduling~(ETGS) mechanism,  we will consider that all the time-varying coefficients are constant at the time interval $[t_j,t_{j+1})$, their values being set to their respective values at time $t_j$. We finally apply neural operators to compute the backstepping kernels by learning the mapping from the system parameters to backstepping kernels.

\section{Preliminaries on boundary control design of linear hyperbolic PDEs with time-independent coefficients,  and Neural operators}\label{sec:Preliminaries}
{
In this section, we consider a nominal case for which the coupling parameters $\sgp(t,x)$ and $\sgn(t,x)$ are time-independent. They will be denoted as $\bar \sigma^+(x)$, $\bar \sigma^-(x)$. Then the nominal system with time-independent system parameters rewrites
\begin{align}
 \partial_t u_{\text{nom}}(t, x)+\lambda(x) \partial_x u_{\text{nom}}(t, x)&=\bar \sigma^{+}(x) v_{\text{nom}}(t, x),\label{eq:nomsys1}\\
 \partial_t v_{\text{nom}}(t, x)-\mu(x) \partial_x v_{\text{nom}}(t, x)&=\bar\sigma^{-}(x) u_{\text{nom}}(t, x),
\end{align}
with boundary conditions
\begin{align}
 u_{\text{nom}}(t, 0)&=q v_{\text{nom}}(t, 0), \\
 v_{\text{nom}}(t, 1)&=\rho u_{\text{nom}}(t, 1)+U_{\text{nom}}(t). \label{eq:nomsys4}
\end{align}
This nominal case will give us the opportunity to introduce key concepts and results for the time-varying case.}

\subsection{Nominal controller design}
We now recall the nominal stabilizing backstepping control law defined in~\cite{vazquez_backstepping_2011}. This control law is based on the backstepping transformation:
\begin{align}
    \alpha_{\text{nom}}(t,x)
    &= u_{\text{nom}}(t,x)
    + \int_0^x \Kuu_{\text{nom}}(x,\xi)u_{\text{nom}}(t,\xi)\,d\xi
    + \int_0^x \Kuv_{\text{nom}}(x,\xi)v_{\text{nom}}(t,\xi)\,d\xi,\label{eq:nominaltrans1}\\
    \beta_{\text{nom}}(t,x)
    &= v_{\text{nom}}(t,x)
    + \int_0^x \Kvu_{\text{nom}}(x,\xi)u_{\text{nom}}(t,\xi)\,d\xi
    + \int_0^x \Kvv_{\text{nom}}(x,\xi)v_{\text{nom}}(t,\xi)\,d\xi.\label{eq:nominaltrans2}
\end{align}
where $K^{\cdot\cdot}_{\text{nom}}$ are nominal backstepping kernels defined on the triangular domain $\mathcal{D}= \{0 \leq \xi \leq x \leq 1\}$. They are continuous function that satisfy~\cite{vazquez_backstepping_2011} 
\begin{align}
    \Delta(x) \partial_x K_{\text{nom}}(x,\xi) + \partial_\xi (K_{\text{nom}}(x,\xi)\Delta(\xi)) = K_{\text{nom}} \Sigma(\xi) \label{eq:nomker1},
\end{align}
with boundary conditions
\begin{align}
    \Delta(x) K_{\text{nom}}(x,x) - K_{\text{nom}}(x,x)\Delta(x) &= -\Sigma(x),\\
    K_{\text{nom}}^{uu}(x,0) &= \frac{\mu(0)}{\lambda(0) q}K_{\text{nom}}^{uv}(x,0),\\
    K^{vv}_{\text{nom}}(x,0) &= \frac{\lambda(0)q}{\mu(0)}K_{\text{nom}}^{vu}(x,0),\label{eq:nomker8}
\end{align}
where $K_{\text{nom}}(x,y):= \begin{pmatrix}
    \Kuu_{\text{nom}}(x,\xi) & \Kuv_{\text{nom}}(x,\xi)\\
    \Kvu_{\text{nom}}(x,\xi) & \Kvv_{\text{nom}}(x,\xi)
\end{pmatrix}$, $\Delta(x) := \begin{pmatrix}
    \lambda(x) & 0\\
    0 & \mu(x)
\end{pmatrix}$, and $\Sigma(x):=\begin{pmatrix}
    0 & \bar\sigma^+(x)\\
    \bar\sigma^-(x) & 0
\end{pmatrix}$.
The nominal control law is defined as:
\begin{equation}\label{eq:TimeIndControl}
\begin{aligned}
    U_{\text{nom}}(t) ={}& \int_0^1 (\rho\Kuu_{\text{nom}}(1,\xi)-\Kvu_{\text{nom}}(1,\xi))u_{\text{nom}}(t,\xi)\,d\xi
    + \int_0^1 (\rho\Kuv_{\text{nom}}(1,\xi)-\Kvv_{\text{nom}}(1,\xi))v_{\text{nom}}(t,\xi)\,d\xi.
\end{aligned}
\end{equation}
The transformation \eqref{eq:nominaltrans1}-\eqref{eq:nominaltrans2} with the control law~\eqref{eq:TimeIndControl} maps the system~\eqref{eq:nominaltrans1}-\eqref{eq:nominaltrans2} to the target system as
\begin{align}
    \partial_t \alpha_{\text{nom}}(t, x) + \lambda(x) \partial_x \alpha_{\text{nom}}(t, x) &= 0,\\
    \partial_t \beta_{\text{nom}}(t, x) - \mu(x) \partial_x \beta_{\text{nom}}(t, x) &= 0,\\
    \alpha_{\text{nom}}(t,0) &= q\beta_{\text{nom}}(t,0),\\
    \beta_{\text{nom}}(t,1) &= \rho \alpha_{\text{nom}}(t,1).
\end{align}
This target system is exponentially stable since $|\rho q|<1$. Therefore, due to the boundedness and invertibility of the Volterra transformation \eqref{eq:nominaltrans1}-\eqref{eq:nominaltrans2}, the nominal system~\eqref{eq:nomsys1}-\eqref{eq:nomsys4} is exponential stable in the sense of the $L^2$-norm under the nominal control law~\eqref{eq:TimeIndControl}~\cite[Lemma 2]{AuriolDelayRobustControl2018}.
\subsection{Properties of the kernels}
\begin{prop}[\cite{wang_backstepping_2025,di2013stabilization}]\label{lm:boundK}
    For every $\lambda,\mu\in \mathcal{C}^1([0,1])$, $\bar \sigma^+, \bar \sigma^- \in$ $\mathcal{C}^1([t_0,\infty]\times [0,1])$ $\cap$ $L^\infty([t_0,\infty)\times [0,1])$, and $q,\rho\in\mathbb{R}$, There exists $M_1>0$ such that for all nominal backstepping kernels $K^{\cdot\cdot}_{\text{nom}}$ satisfying the kernel equations~\eqref{eq:nomker1}-\eqref{eq:nomker8} admit a unique solution with the following bound
    \begin{align}
    \abs{K^{\cdot\cdot}_{\text{nom}}(x,\xi)}\le M_1, \forall(x,\xi)\in \mathcal{D}.
    \end{align}
\end{prop}
For the nominal backstepping transformation, we denote its compact form as $(\alpha_{\text{nom}}(t,x),\beta_{\text{nom}}(t,x))
    := (\mathcal{T}(u_{\text{nom}}(t,\cdot),v_{\text{nom}}(t,\cdot)))(x)$,
and its inverse transformation as $(u_{\text{nom}}(t,x),v_{\text{nom}}(t,x))
    := (\mathcal{T}^{-1}(\alpha_{\text{nom}}(t,\cdot),\beta_{\text{nom}}(t,\cdot)))(x)$.
We have the following proposition.
\begin{prop}[\cite{vazquez_backstepping_2011}]\label{lm:uniform_T}
There exists a constant $M_2\ge 1$, such that for all
$(u,v)\in(L^2([0,1]))^2$,
\begin{align}
\label{eq:uniform_T_bound}
\begin{aligned}
\norm{\mathcal{T} (u_{\text{nom}},v_{\text{nom}})}
&\le M_2 \norm{u_{\text{nom}}(t,\cdot),v_{\text{nom}}(t,\cdot)},\\
\norm{\mathcal{T}^{-1}(\alpha_{\text{nom}},\beta_{\text{nom}})}
&\le M_2 \norm{\alpha_{\text{nom}}(t,\cdot),\beta_{\text{nom}}(t,\cdot)}.
\end{aligned}
\end{align}
\end{prop}

\subsection{Preliminaries of neural operators}\label{section:preliminaries_NO}
Neural operators are a class of machine learning models designed to learn mappings between infinite-dimensional function spaces~\cite{lu2021learning}.
They extend the concept of neural networks to operate on functions rather than finite-dimensional vectors, making them particularly suitable for approximating solution operators of PDEs and other functional relationships. {For a general operator mapping, neural operators can be applied to approximate the mapping using suitable neural network size parameterized by learned parameter weights. Based on the universal approximation theorem~\cite{bhan2023neural,chen1995universal}, we recall the following proposition}
\begin{prop}\label{Deeptheo}
    Let $X \subset \mathbb{R}^{d_x}$, $Y$ $\subset$ $\mathbb{R}^{d_y}$ be compact sets of vectors $\mathbf{x}\in X$ and $\mathbf{y}\in Y$. Let $\mathcal{U}$: $X \rightarrow \mathbb{U} \subset \mathbb{R}^{d_u}$ and $\mathcal{V}$: $Y \rightarrow \mathbb{V} \subset \mathbb{R}^{d_v}$ be sets of continuous functions $\mathbf{u}(\mathbf{x})$ and $\mathbf{v}(\mathbf{y})$, respectively. Let $\mathcal{U}$ be also compact. Assume the operator $\mathcal{G}$: $\mathcal{U} \rightarrow \mathcal{V}$ is continuous. Then, for all $\epsilon > 0$, there exists a $m^\star,p^\star \in \mathbb{N}$ such that for each $m \geq m^\star$, $p \geq p^\star$, there exist {trainable parameters} $\theta^{(\mathbf{k})}$, $\vartheta^{(\mathbf{k})}$, neural networks $f^{\mathcal{N}}\left(\cdot ; \theta^{(\mathbf{k})}\right), g^{\mathcal{N}}\left(\cdot ; \vartheta^{(\mathbf{k})}\right), \mathbf{k}= 1,\dots,p$ and $\mathbf{x}_j \in X, j=1, \ldots, m$, with corresponding $\mathbf{u}_m=\left(\mathbf{u}\left(\mathbf{x}_1\right), \mathbf{u}\left(\mathbf{x}_2\right), \cdots, \mathbf{u}\left(\mathbf{x}_m\right)\right)^{\top}$, such that
    \begin{align}
        \sup_{\mathbf{u}\in \mathcal{U}}\sup_{\mathbf{y}\in Y }
        \left|\mathcal{G}(\mathbf{u})(\mathbf{y})-\mathcal{G}_{\mathcal{N}}\left(\mathbf{u}_m\right)(\mathbf{y})\right|
        &<\epsilon,
    \end{align}
    for all functions $\mathbf{u} \in \mathcal{U}$ and all $\mathbf{y}\in Y$ of $\mathcal{G}(\mathbf{u})(\mathbf{y}) \in \mathcal{V}$.
\end{prop}
Building on these foundational results, several applications of neural operators for PDE control have been developed. These include the approximation of backstepping gain kernels and control laws~\cite{bhan2023neural}, as well as observer gains~\cite{KrsticNeuralOperators2024}. The framework was later extended to backstepping kernels for hyperbolic PDEs~\cite{wang_backstepping_2025} and integrated into gain-scheduling methods~\cite{lamarque_gain_2025}.

\subsection{Existence and continuity of kernel operator}
{ 
For the nominal system \eqref{eq:nomsys1}-\eqref{eq:nomsys4}, we fix a compact admissible parameter set $\mathcal S$ contained in the following bounded admissible set $\mathcal A$:
\begin{equation}\label{eq:admissible_kernel_set}
\begin{aligned}
\mathcal A := \big\{&(\lambda,\mu,\bar\sigma^+,\bar\sigma^-,q,\rho):
\lambda,\mu,\bar\sigma^+,\bar\sigma^-\in \mathcal C^1([0,1]),\
\underline{\lambda}\leq\lambda(x)\leq\bar \lambda,\
\underline{\mu}\leq\mu(x)\leq\bar \mu,\\
&\|\lambda\|_{\mathcal C^1}+\|\mu\|_{\mathcal C^1}\le M_{\lambda\mu},\
\|\bar\sigma^+\|_{\mathcal C^1}+\|\bar\sigma^-\|_{\mathcal C^1}\le M_{\sigma,1},\
\abs{\bar\sigma^{\pm}(x)}\le M_\sigma,\\
&0<\underline q\le |q|\leq \bar q,\
\abs{\rho}\leq \bar \rho,\
|\rho q|\le 1-\eta
\big\},
\end{aligned}
\end{equation}
where $0 < \underline{\lambda} \leq \bar \lambda$, $0 < \underline{\mu} \leq \bar \mu$, $\underline q>0$, $\eta\in(0,1)$, and the remaining constants are positive. The compactness of $\mathcal S\subset\mathcal A$ is an admissibility requirement on the scheduling family used for operator learning. For instance, it is satisfied when the admissible coefficients are generated by a compact finite-dimensional parametrization; it is not a consequence of boundedness in $\mathcal C^1([0,1])$ alone.
}
We have the following lemma.
\begin{lem}\label{lem:kernel_operator_lipschitz}
For \(\Phi :=(\lambda,\mu,\bar\sigma^+,\bar\sigma^-,q,\rho)\in\mathcal S\), the kernel operator $\mathcal{K}:\mathcal{S}\to \big(\mathcal{C}^1(\mathcal{D})\big)^4$ is defined by
\begin{equation*}
    \mathcal{K}(\Phi)(x,\xi)
    := \big(\Kuu_{\text{nom}}(x,\xi),\Kuv_{\text{nom}}(x,\xi),\Kvu_{\text{nom}}(x,\xi),\Kvv_{\text{nom}}(x,\xi)\big).
\end{equation*}
It is locally Lipschitz. More precisely, there exists a constant $c_1>0$ such that for all
$\Phi_a, \Phi_b \in\mathcal{S}$,
\begin{align}\label{kernelop}
\norm{\mathcal K(\Phi_{a})-\mathcal K(\Phi_{b})}_{(\mathcal{C}^1(\mathcal{D}))^4}
&\leq c_1 \norm{\Phi_{a}-\Phi_{b}}_{\mathcal{S}},
\end{align}
where $\norm{\cdot}_{\mathcal{C}^1(\mathcal{D})}
    := \norm{\cdot}_{L^\infty(\mathcal{D})}
    + \norm{\partial_x(\cdot)}_{L^\infty(\mathcal{D})}
    + \norm{\partial_\xi(\cdot)}_{L^\infty(\mathcal{D})}$, $\norm{\Phi_a-\Phi_b}_{\mathcal{S}}
    := \norm{\lambda_a-\lambda_b}_{\mathcal{C}^1}
    + \norm{\mu_a-\mu_b}_{\mathcal{C}^1}
    + \norm{\bar\sigma^+_a-\bar\sigma^+_b}_{\mathcal{C}^1}
    + \norm{\bar\sigma^-_a-\bar\sigma^-_b}_{\mathcal{C}^1}
    + \abs{q_a-q_b}
    + \abs{\rho_a-\rho_b}$.
\end{lem}
\begin{pf}
{ 
The proof is divided into four steps. Throughout the proof, write
$K^1=K^{uu}$, $K^2=K^{uv}$, $K^3=K^{vu}$, and $K^4=K^{vv}$ for the four scalar
kernel components.

\emph{Step 1: uniform kernel solvability on the admissible set.}
For each $\Phi\in\mathcal S$, the backstepping kernel equations
\eqref{eq:nomker1}--\eqref{eq:nomker8} admit a unique solution in
$(\mathcal C^1(\mathcal D))^4$ by the standard kernel well-posedness result
for $2\times2$ heterodirectional hyperbolic systems~\cite{vazquez_backstepping_2011}.
The characteristic construction depends on the transport speeds, the
in-domain coefficients, and the boundary parameter $q$ appearing in
\eqref{eq:nomker8}. Since $\mathcal S\subset\mathcal A$ imposes uniform lower
and upper velocity bounds, uniform $\mathcal C^1$ bounds on
$\lambda,\mu,\bar\sigma^+,\bar\sigma^-$, and $|q|\ge \underline q>0$, the
characteristic times, the boundary traces, and the coefficients in the
Volterra integral representation of the kernels are bounded by constants
independent of the particular choice of $\Phi$. Hence there exists $M_S>0$,
depending only on the constants defining $\mathcal A$ and on the compact set
$\mathcal S$, such that
\begin{align*}
\|\mathcal K(\Phi)\|_{(\mathcal C^1(\mathcal D))^4}\le M_S,
\qquad \Phi\in\mathcal S .
\end{align*}

\emph{Step 2: difference system.}
Let
\begin{align*}
\Phi_a=(\lambda_a,\mu_a,\bar\sigma_a^+,\bar\sigma_a^-,q_a,\rho_a),
\qquad
\Phi_b=(\lambda_b,\mu_b,\bar\sigma_b^+,\bar\sigma_b^-,q_b,\rho_b),
\end{align*}
and denote their kernels by $\mathcal K_a$ and $\mathcal K_b$. Set
\begin{align*}
\Delta\mathcal K:=\mathcal K_a-\mathcal K_b,\qquad
\Delta\Phi:=\Phi_a-\Phi_b .
\end{align*}
Subtracting the two kernel systems gives a linear first-order hyperbolic
system for $\Delta\mathcal K$ on $\mathcal D$. The boundary data of this
difference system are differences of the boundary data associated with
$\Phi_a$ and $\Phi_b$. The diagonal boundary data depend on
$\lambda,\mu,\bar\sigma^\pm$, and the boundary data on $\xi=0$ also contain
the factors $q^{-1}$ and $q$. These dependences are smooth on
$\mathcal A$ because the velocities are bounded away from zero and
$|q|\ge \underline q$. Therefore there is a constant $B_\partial>0$ such that
\begin{align*}
\|\Delta\mathcal K\|_{L^\infty(\partial\mathcal D)}
\le B_\partial\|\Delta\Phi\|_{\mathcal S}.
\end{align*}
To make the subtraction precise, use the kernel system with the coefficients
of $\Phi_a$ as the reference system. All terms containing
$\Delta\mathcal K$ are kept on the left-hand side and all remaining terms are
put into a source. Componentwise, each equation has the form
\begin{align*}
\mathscr L_a^m(\Delta K^m)(x,\xi)
=\sum_{n=1}^4 a_{mn}(x,\xi)\Delta K^n(x,\xi)+F_{ab}^m(x,\xi),
\qquad m=1,\ldots,4 .
\end{align*}
Here $\mathscr L_a^m$ denotes the first-order characteristic operator
obtained from the principal transport part of the $m$th scalar kernel equation
with the parameters $\Phi_a$. The coefficients $a_{mn}$ are the zero-order
coefficients multiplying the unknown differences $\Delta K^n$ after
subtraction. They are explicit combinations of
$\lambda_a,\mu_a,\lambda_a',\mu_a',\bar\sigma_a^+,\bar\sigma_a^-$, and are
therefore uniformly bounded on $\mathcal S$. The source term $F_{ab}^m$
contains only coefficient differences multiplying the already known kernels
associated with $\Phi_b$, for example
$(\lambda_a-\lambda_b)\partial_xK_b^n$,
$(\mu_a-\mu_b)\partial_\xi K_b^n$,
$(\lambda_a'-\lambda_b')K_b^n$,
$(\mu_a'-\mu_b')K_b^n$, and
$(\bar\sigma_a^\pm-\bar\sigma_b^\pm)K_b^n$. Since $\mathcal K_b$ is uniformly
bounded in $(\mathcal C^1(\mathcal D))^4$, there is a constant $B_f>0$,
depending only on $\mathcal S$, such that
\begin{align*}
\max_{1\le m\le4}\|F_{ab}^m\|_{\mathcal C^0(\mathcal D)}
\le B_f\|\Delta\Phi\|_{\mathcal S}.
\end{align*}

\emph{Step 3: Volterra estimate for the zeroth-order part.}
Solving each component equation along its characteristic gives a Volterra
representation of the form
\begin{align*}
\Delta K^m(x,\xi)
=\Delta K_\partial^m(x,\xi)
+\int_{\Gamma_m(x,\xi)}
\sum_{n=1}^4 a_{mn}(X_m(s),\Xi_m(s))
\Delta K^n(X_m(s),\Xi_m(s))\,ds
+\int_{\Gamma_m(x,\xi)}
F_{ab}^m(X_m(s),\Xi_m(s))\,ds .
\end{align*}
Here $\Delta K_\partial^m$ denotes the contribution of the boundary data
estimated above, and $\Gamma_m(x,\xi)$ is the characteristic segment for the
$m$th kernel component ending at $(x,\xi)$. The functions
$(X_m(s),\Xi_m(s))$ parametrize this segment. The uniform velocity bounds imply
that the lengths of these characteristic segments are uniformly bounded on
$\mathcal D$. For
$s\in[0,1]$, define the truncated Volterra triangle
\begin{align*}
\mathcal D_s:=\{(x,\xi)\in\mathcal D:\ 0\le \xi\le x\le s\},
\qquad
Y_0(s):=\|\Delta\mathcal K\|_{(L^\infty(\mathcal D_s))^4}.
\end{align*}
Taking componentwise absolute values and using the bounds above gives
\begin{align*}
Y_0(s)
\le C_0\|\Delta\Phi\|_{\mathcal S}
+C_1\int_0^sY_0(r)\,dr ,
\qquad s\in[0,1].
\end{align*}
The Volterra--Gronwall inequality then yields
\begin{align*}
\|\Delta\mathcal K\|_{(L^\infty(\mathcal D))^4}
\le C_0e^{C_1}\|\Delta\Phi\|_{\mathcal S}.
\end{align*}

\emph{Step 4: derivative estimates.}
It remains to estimate the first derivatives. In the standard characteristic
construction of the kernel equations, the first derivatives
$\partial_x K^m$ and $\partial_\xi K^m$ satisfy characteristic integral
equations obtained from the differentiated kernel system. Applying the same
subtraction argument to these integral equations gives equations for
$\partial_x\Delta K^m$ and $\partial_\xi\Delta K^m$. Their coefficients are
uniformly bounded because the velocities are separated from zero and
$\lambda,\mu,\bar\sigma^+,\bar\sigma^-$ have uniformly bounded
$\mathcal C^1$ norms on $\mathcal S$. Their inhomogeneous terms are bounded by
$C\|\Delta\Phi\|_{\mathcal S}$ plus lower-order terms already controlled by
the estimate for $Y_0$. Thus, with
\begin{align*}
Y_1(s):=\|\Delta\mathcal K\|_{(\mathcal C^1(\mathcal D_s))^4},
\end{align*}
the differentiated Volterra equations give
\begin{align*}
Y_1(s)
\le C_2\|\Delta\Phi\|_{\mathcal S}
+C_3\int_0^sY_1(r)\,dr ,
\qquad s\in[0,1],
\end{align*}
where $C_2,C_3>0$ depend only on $\mathcal S$. Applying Volterra--Gronwall
once more gives
\begin{align*}
Y_1(1)\le C_2e^{C_3}\|\Delta\Phi\|_{\mathcal S}.
\end{align*}
Therefore \eqref{kernelop} holds with $c_1:=C_2e^{C_3}$. This finishes the
proof of Lemma~\ref{lem:kernel_operator_lipschitz}.
}
\end{pf}
\begin{rem}
    The local Lipschitz continuity of $\mathcal K(\Phi)(x,\xi)$ established above implies that, on the fixed compact admissible set $\mathcal{S}$, $\mathcal K$ is continuous as a map into $(\mathcal{C}^1(\mathcal D))^4$.  Consequently, universal approximation results for operator-learning architectures (e.g., DeepONet), in the corresponding uniform operator topology, guarantee the existence of a neural operator to approximate the mapping.
\end{rem}
By using the neural operator, we obtain that for the compact admissible parameter set $\mathcal{S}$, there exists a neural operator $\widehat{\mathcal{K}}: \mathcal{S} \to (\mathcal{C}^1(\mathcal D))^4$, such that 
\begin{align}
    \sup_{\Phi\in\mathcal{S}} \norm{\mathcal{K}(\Phi) - \widehat{\mathcal{K}}(\Phi)}_{(\mathcal C^1(\mathcal D))^4} \leq \epsilon.
\end{align}
Equivalently, for every $\Phi=(\lambda(x),\mu(x), \bar\sigma^+,\bar\sigma^-, q, \rho)$, if we denote $\widehat{\mathcal{K}}(\Phi):=(\hKuu_{\text{nom}}, \hKuv_{\text{nom}}, \hKvu_{\text{nom}}, \hKvv_{\text{nom}})$, then the approximation errors $\widetilde{K}_{\text{nom}}^{\cdot\cdot} = K^{\cdot\cdot}_{\text{nom}} - \widehat{K}^{\cdot\cdot}_{\text{nom}}$ satisfy 
\begin{align}\label{eq_approx_error}
\|\widetilde{K}_{\text{nom}}^{\cdot\cdot}\|_{\mathcal{C}^1(\mathcal D)} \leq \epsilon.
\end{align}

By using the neural operator to approximate the mapping from system parameters to nominal kernels, the nominal system~\eqref{eq:nomsys1}-\eqref{eq:nomsys4} achieves exponential stability under the approximated nominal control law. This law is obtained by replacing the kernels in the original nominal control law~\eqref{eq:TimeIndControl}, as it has been done in~\cite{wang_backstepping_2025}. Even for PDEs with Markov-jumping parameters, mean-square exponential stability can be achieved, as demonstrated in previous work~\cite{zhang_operator_2025}.

\section{System (1)--(4) under event-triggered gain-scheduling mechanism}\label{sec:timevary}
Let us come back to system \eqref{eq:sys1}-\eqref{eq:sys4} involving  time- and space-varying coupling terms $\sgp(t,x)$ and $\sgn(t,x)$, which leads to a more complex  dynamics and controller design. To stabilize that system, the controller design incorporates an event-triggered gain scheduling method \cite{espitia_event-triggered_2025} together with neural operators with the aim to accelerate the control gain computation at each triggering time.
\subsection{Backstepping design for original system} \label{Sec_BS}
Let $\{t_j\},j\in \mathbb{N}$ be an increasing sequence of times. {We first denote the error when sampling as:}
\begin{align}
    \delsgp(t,x)&= \sgp(t,x)-  \sgp(t_j,x),\\
    \delsgn(t,x)&= \sgn(t,x)- \sgn(t_j,x).
\end{align}
{We also define the sampled version of different coupling coefficients in~\eqref{eq:sys1}-\eqref{eq:sys4} as $\sgp_j(x):= \sgp(t_j,x), \sgn_j(x) := \sgn(t_j,x)$.}
Therefore, at each triggered event, the original system~\eqref{eq:sys1}-\eqref{eq:sys4} rewrites as
\begin{align}
 \partial_t u(t,x)+\lambda(x) \partial_x u(t,x)&=\sgp_j(x)v(t,x) + \delsgp v(t,x),\label{eq:sys1del}\\
 \partial_t v(t,x)-\mu(x) \partial_x v(t,x)&=\sgn_j(x) u(t,x) + \delsgn u(t,x),
\end{align}
with boundary conditions
\begin{align}
 u(t,0)&=q v(0, t), \label{eq:sys3del}\\
 v(t,1)&=\rho u(1, t)+U(t). \label{eq:sys4del}
\end{align}
Consider the following backstepping transformation defined at each time interval $t\in [t_j,t_{j+1})$ and for all $x\in[0, 1]$ as:
\begin{align}
    \alpha_j(t,x)
    &= u(t,x)
    + \int_0^x \Kuu_j(x,\xi)u(t,\xi)\,d\xi
    + \int_0^x \Kuv_j(x,\xi)v(t,\xi)\,d\xi,\label{eq:trans1}\\
    \beta_j(t,x)
    &= v(t,x)
    + \int_0^x \Kvu_j(x,\xi)u(t,\xi)\,d\xi
    + \int_0^x \Kvv_j(x,\xi)v(t,\xi)\,d\xi.\label{eq:trans2}
\end{align}
which can be rewritten as in the compact form $(\alpha_j(t,x),\beta_j(t,x)) := (\mathcal{T}_j(u(t,\cdot),v(t,\cdot)))(x)$.
 The kernels $K^{\cdot\cdot}_j$ are continuous function defined on the same triangular domain $\mathcal{D}$.  They satisfy analogous equations to~\eqref{eq:nomker1}-\eqref{eq:nomker8}.

{The transformation~\eqref{eq:trans1}-\eqref{eq:trans2} is a Volterra transformation and is invertible. Consequently, there exist bounded functions $L^{\cdot\cdot}_j$ defined on the same triangular domain $\mathcal{D}$ such that for all $t\in [t_j,t_{j+1})$ and $x\in[0, 1]$ we have}
\begin{align}
    u(t,x)
    &= \alpha_j(t,x)
    + \int_0^x \Laa_j(x,\xi)\alpha_j(t,\xi)\,d\xi
    + \int_0^x \Lab_j(x,\xi)\beta_j(t,\xi)\,d\xi,\\
    v(t,x)
    &= \beta_j(t,x)
    + \int_0^x \Lba_j(x,\xi)\alpha_j(t,\xi)\,d\xi
    + \int_0^x \Lbb_j(x,\xi)\beta_j(t,\xi)\,d\xi.
\end{align}
{The inverse transformation can be written in compact form as $(u(t,x),v(t,x)) = (\mathcal{T}_j^{-1}(\alpha_j(t,\cdot),\beta_j(t,\cdot)))(x)$.}
For all $t\in[t_j,t_{j+1})$, we choose the boundary control input as
\begin{equation}\label{eq:ControOri}
\begin{aligned}
    U(t) ={}& \int_0^1 (\rho\Kuu_j(1,\xi)-\Kvu_j(1,\xi))u(t,\xi)\,d\xi
    + \int_0^1 (\rho\Kuv_j(1,\xi)-\Kvv_j(1,\xi))v(t,\xi)\,d\xi.
\end{aligned}
\end{equation}
In~\cite{espitia_event-triggered_2025}, it was demonstrated that this boundary control law~\eqref{eq:ControOri} exponentially stabilizes the system~\eqref{eq:sys1del}--\eqref{eq:sys4del}, provided that the sequence of triggering events is appropriately designed. Additionally, it was shown that the triggering mechanism avoids Zeno behavior, ensuring that no infinite number of triggering events occurs within any finite time interval. However, a significant computational challenge remains: the backstepping kernels must be recalculated at each triggering event, which can lead to a prohibitive computational burden. This is why we propose leveraging neural operators to approximate these kernels, thereby accelerating the computational process while preserving stability guarantees.

\subsection{Neural operators-approximated control law}
In this part, we will utilize the neural operators to approximate the mapping from the set of parameters in $\mathcal{S}$ to the backstepping kernels at each triggered-event. Let $\Phi_j := (\lambda(x),\mu(x), \sgp_j(x),\sgn_j(x), q, \rho)$ denote the system parameters at the triggered-event, then the corresponding backstepping kernels are denoted by $\mathcal{K}(\Phi_j)(x,\xi)$.
In what follows, we will denote with a superscript $\widehat\cdot$ the approximated kernels obtained using the neural operator at the triggered time. We will denote with a superscript $\widetilde \cdot$, the corresponding error kernels. More precisely, for any set of parameters in $\mathcal{S}$, we denote $\widehat{K}_j^{\cdot\cdot}$ the approximation of the backstepping kernels obtained using neural operators. The corresponding error kernels satisfies
\begin{align}
\widetilde K_j^{\cdot\cdot}(x,\xi):=K_j^{\cdot\cdot}(x,\xi)-\widehat K_j^{\cdot\cdot}(x,\xi).
\end{align}
Due to~\eqref{eq_approx_error}, we have
 \begin{align}\label{eq:kernel_approximation_bound}
\|\widetilde K_j^{\cdot\cdot}\|_{\mathcal C^1(\mathcal D)}\le \epsilon.
\end{align}
These approximated kernels can then be used to compute the NO-approximated control law.
Hence, the boundary value of the state in \eqref{eq:sys4}  (also in \eqref{eq:sys4del}) is modified and now given by
\begin{equation}
 v(t,1)=\rho u(1, t)+\widehat{U}(t). \label{eq:sys4del_withNO}
\end{equation}
where the approximated boundary control is given as follows:
\begin{equation}\label{eq:nocontrolaw}
\begin{aligned}
    \widehat{U}(t) ={}& \int_0^1 (\rho\hKuu_j(1,\xi)-\hKvu_j(1,\xi))u(t,\xi)\,d\xi
    + \int_0^1 (\rho\hKuv_j(1,\xi)-\hKvv_j(1,\xi))v(t,\xi)\,d\xi.
\end{aligned}
\end{equation}
for all $t\in[t_{j},t_{j+1})$.
\subsection{Target system}
Consider now the system~\eqref{eq:sys1del}-\eqref{eq:sys3del}, \eqref{eq:sys4del_withNO} with the NO-approximated control law~\eqref{eq:nocontrolaw}. Applying the backstepping transformations \eqref{eq:trans1}-\eqref{eq:trans2}, we obtain the following target system
\begin{align}
    \partial_t\alpha_j(t,x)+ \lambda(x) \partial_x\alpha_j(t,x)&= \heone(t,x) + r_{j,1}^\epsilon(t,x),\label{eq:tar1}\\
    \partial_t\beta_j(t,x)- \mu(x) \partial_x\beta_j(t,x)&= \hetwo(t,x) + \rtwo(t,x),\label{eq:tar2}
\end{align}
with boundary conditions
\begin{align}
    \alpha_j(t,0) &= q\beta_j(t,0),\label{eq:tar3}\\
    \beta_j(t,1) &= \rho \alpha_j(t,1)
    + \int_0^1 \big(\tKvu_j(1,\xi)-\rho \tKuu_j(1,\xi)\big)u(t,\xi)\,d\xi + \int_0^1 \big(\tKvv_j(1,\xi)-\rho \tKuv_j(1,\xi)\big)v(t,\xi)\,d\xi. \label{eq:tar4}
\end{align}
The perturbation terms $\heone$ and $\hetwo$ are defined as
\begin{align}
    \heone(t,x)
    &= \delta_{\sigma_j^+}(t,x)v
    + \int_0^x \widehat K_j^{uu}(x,\xi)\delta_{\sigma_j^+}(t,\xi)v(t,\xi)\,d\xi  + \int_0^x \widehat K_j^{uv}(x,\xi)\delta_{\sigma_j^-}(t,\xi)u(t,\xi)\,d\xi,\label{eq:e1}\\
    \hetwo(t,x)
    &= \delta_{\sigma_j^-}(t,x)u
    + \int_0^x \widehat K_j^{vu}(x,\xi)\delta_{\sigma_j^+}(t,\xi)v(t,\xi)\,d\xi  + \int_0^x \widehat K_j^{vv}(x,\xi)\delta_{\sigma_j^-}(t,\xi)u(t,\xi)\,d\xi.\label{eq:e2}
\end{align}
and the approximation error terms are defined as
\begin{align}
    \rone(t,x)
    &= \int_0^x \tKuu_j(x,\xi)\delta_{\sigma_j^+}(t,\xi)v(t,\xi)\,d\xi
    + \int_0^x \tKuv_j(x,\xi)\delta_{\sigma_j^-}(t,\xi)u(t,\xi)\,d\xi,\\
    \rtwo(t,x)
    &= \int_0^x \tKvu_j(x,\xi)\delta_{\sigma_j^+}(t,\xi)v(t,\xi)\,d\xi
    + \int_0^x \tKvv_j(x,\xi)\delta_{\sigma_j^-}(t,\xi)u(t,\xi)\,d\xi.
\end{align}
    In order to obtain an implementable triggering condition in the event-triggered gain-scheduling mechanism, we use NO-approximated kernels $\widehat{K}^{\cdot\cdot}$ and the approximation error $\epsilon$ to denote the right-hand side of target system~\eqref{eq:tar1}-\eqref{eq:tar2}.
\begin{rem}
    The target system~\eqref{eq:tar1}-\eqref{eq:tar4} exhibits perturbations in both the domain and boundary conditions, arising from the time-varying nature of the coupling terms and the neural operator approximation errors.
    The control law~\eqref{eq:ControOri} exponentially stabilizes the system in the sense of the $L^2$-norm in the absence of these perturbations from indomain coupling terms $\sgp(t,x),\sgn(t,x)$ and approximation error $\epsilon$.
    These perturbations can be effectively managed using Lyapunov-based techniques, as demonstrated in~\cite{espitia_event-triggered_2025}.
    The neural operator approximation error $\epsilon$ can be made arbitrarily small by choosing the suitable size of the neural operator (e.g., number of layers, neurons per layer) and training data.
\end{rem}

\subsection{Well-posedness issues}\label{subsec:wp}
We now address the well-posedness of the closed-loop system~\eqref{eq:sys1}--\eqref{eq:sys4} under the NO-approximated boundary control law \eqref{eq:nocontrolaw}. Since the controller gains are updated only at triggering times, the closed-loop system is piecewise defined over the inter-event intervals. We have the following lemma.
\begin{lem}
\label{lem:wp}
For any triggering sequence $\{t_j\}_{j\in\mathbb N}$ generated by the event-triggering mechanism, and for every initial condition $(u_0(x), v_0(x))\in (L^2([0,1]))^2$, the closed-loop system \eqref{eq:sys1}--\eqref{eq:sys4} with the NO-approximated control law \eqref{eq:nocontrolaw} admits a unique solution $(u,v)\in \mathcal C^0\big([0,\lim_{j\to \infty}(t_j));(L^2([0,1]))^2\big)$. 
\end{lem}

\begin{pf}
{ 
The proof is divided into four steps. Let
\begin{align*}
T_\ast:=\lim_{j\to\infty}t_j .
\end{align*}

\emph{Step 1: frozen feedback form on an inter-event interval.}
Fix $j$ and consider $t\in[t_j,t_{j+1})$. On this interval, the kernels and
controller are frozen at $t_j$. The NO-approximated boundary input can be
rewritten as
\begin{equation}
\begin{aligned}
    \widehat{U}(t)
    &= \int_0^1 F_{1,j}(\xi)u(t,\xi)\,d\xi
    + \int_0^1 F_{2,j}(\xi)v(t,\xi)\,d\xi,
\end{aligned}
\end{equation}
where
\begin{align*}
F_{1,j}(\xi)=\rho\hKuu_j(1,\xi)-\hKvu_j(1,\xi),
\qquad
F_{2,j}(\xi)=\rho\hKuv_j(1,\xi)-\hKvv_j(1,\xi).
\end{align*}

\emph{Step 2: boundedness of the boundary feedback gains.}
By Proposition~\ref{lm:boundK} and the approximation bound
\eqref{eq:kernel_approximation_bound}, the approximated kernels satisfy
\begin{align*}
\max_{\star\in\{uu,uv,vu,vv\}}
\|\widehat K_j^\star\|_{L^\infty(\mathcal D)}
\le
\max_{\star}\|K_j^\star\|_{L^\infty(\mathcal D)}
+\max_{\star}\|\widetilde K_j^\star\|_{L^\infty(\mathcal D)}
\le M_1+\epsilon .
\end{align*}
Consequently,
\begin{align*}
|F_{1,j}(\xi)|
\le (1+|\rho|)(M_1+\epsilon),
\qquad
|F_{2,j}(\xi)|
\le (1+|\rho|)(M_1+\epsilon),
\end{align*}
for almost every $\xi\in[0,1]$, uniformly in $j$. Defining the piecewise
functions $F_i(t,\xi):=F_{i,j}(\xi)$ for $t\in[t_j,t_{j+1})$, $i=1,2$, gives
\begin{align*}
F_1,F_2\in L^\infty((0,T_\ast)\times(0,1)).
\end{align*}

\emph{Step 3: interval-wise well-posedness.}
The in-domain coefficients $\sgp,\sgn$ are bounded by
\eqref{eq:assumption_bound_coupling_coeff}. Moreover, the boundary input is a
bounded linear functional of the state. Indeed, by Cauchy--Schwarz,
\begin{align}
|\widehat U(t)|
&\le \|F_{1,j}\|_{L^\infty}\int_0^1|u(t,\xi)|\,d\xi
     +\|F_{2,j}\|_{L^\infty}\int_0^1|v(t,\xi)|\,d\xi \nonumber\\
&\le C_U\|(u(t,\cdot),v(t,\cdot))\|,
\end{align}
where
\begin{align*}
C_U:=2(1+|\rho|)(M_1+\epsilon).
\end{align*}
Thus, on each inter-event interval, the closed-loop plant is a linear
first-order hyperbolic balance law with bounded in-domain coefficients and a
bounded integral boundary feedback. The well-posedness theorem for such
linear hyperbolic systems with bounded boundary feedback, as used in
\cite{CoronBoundaryStabilization2021}, gives a unique solution on
$[t_j,t_{j+1})$ for the initial value inherited at $t_j$.

\emph{Step 4: concatenation across triggering times.}
At a triggering instant the plant state is not reset; only the controller
kernels are updated. Hence
\begin{align*}
(u,v)(t_{j+1}^{+},\cdot)=(u,v)(t_{j+1}^{-},\cdot)
\quad\text{in }(L^2([0,1]))^2 .
\end{align*}
The terminal state on one interval is therefore the initial state on the
next interval. By uniqueness on each interval, these interval-wise solutions
concatenate into a unique solution on $[0,T_\ast)$. This proves
\begin{align*}
(u,v)\in\mathcal C^0\big([0,T_\ast);(L^2([0,1]))^2\big),
\end{align*}
which is the claimed well-posedness result. This finishes the proof of
Lemma~\ref{lem:wp}.
}
\end{pf}

\section{Event-triggered gain scheduling via neural operators}\label{sec:ETGS}
In this section, we will introduce an event-triggering mechanism to determine the triggering times $\{t_j\}$ to update the backstepping kernels using the NO-approximated backstepping kernels.
The triggering condition relies on the evolution of errors when sampling the coupling terms i.e., $\widehat e_j(t,x)= (\heone(t,x),\hetwo(t,x))$ at each triggering time and a {Lyapunov-like functional}
$V((\alpha_j,\beta_j))$ defined as
\begin{align}
    V(t)
    = \int_0^1 \tfrac{\e^{-\nu\phi_1(x)}}{\lambda(x)}\alpha_j^2(t,x)\,dx
    + \int_0^1 a\tfrac{\e^{\nu\phi_2(x)}}{\mu(x)}\beta_j^2(t,x)\,dx := \inner{g(\cdot)\gamma_j(t,\cdot)}{\gamma_j(t,\cdot)},
\end{align}
where $a>0$, $\nu>0$, $\gamma_j(t,x)= (\alpha_j(t,x),\beta_j(t,x))$, and $g(x):= \bigg(\tfrac{\e^{-\nu\phi_1(x)}}{\lambda(x)}, a\tfrac{\e^{\nu\phi_2(x)}}{\mu(x)}\bigg)$, $\phi_1(x):= \int_0^x \tfrac{1}{\lambda(s)}\,ds$, $\phi_2(x):= \int_0^x \tfrac{1}{\mu(s)}\,ds$.
Next, for ease of notation, we use $V(t)=V(\gamma_j(t))=V((\alpha_j(t,\cdot),\beta_j(t,\cdot)))=V(w(t,\cdot))=\inner{g(\cdot)\mathcal{T}_j(w(t,\cdot))}{\mathcal{T}_j(w(t,\cdot))}$, and set $w(t,x):=(u(t,x),v(t,x))$ and $r^\epsilon_j(t,x):=(\rone(t,x),\rtwo(t,x))$.
{Due to the fact that the backstepping transformation is boundedly invertible, the Lyapunov functional $V((\alpha_j,\beta_j))$ is uniformly equivalent to the state norm such that there exist $m_V,M_V>0$, independent of $j$, such that for all $t\in[t_j,t_{j+1})$},
\begin{align}\label{eq:V_equiv}
m_V\norm{w(t,\cdot)}^2 &\le V(t) \le M_V \norm{w(t,\cdot)}^2.
\end{align}
Let $R_L\in (0, 1)$ be a design parameter, and define the following set
\begin{align}\label{eq:etcondition}
    {E_L(t_j) := \{t> t_j:
    2\inner{g(\cdot)\mathcal{T}_j(w(t,\cdot))}{\widehat{e}_j(t,\cdot)} > \nu R_L V(t)\}}
\end{align}
The times of events $t_j\geq 0$ with $t_0=0$ leads to a finite set of times, which is defined as
\begin{enumerate}
    \item if $E_L(t_j) = \emptyset$, then the set of the times of events is $\{t_0,t_1,\cdots,t_j\}$.
    \item if $E_L(t_j) \neq \emptyset$, then the next triggering time is defined as $t_{j+1} = \inf E_L(t_j)$.
\end{enumerate}

\subsection{Avoidance of Zeno phenomenon}
To prevent Zeno behavior, we establish a minimum inter-event time between consecutive triggering times.
This is to guarantee that there are no infinite triggering times that occur in a finite time interval
Specifically, we have the following lemma.
\begin{lem}\label{lem:minidwell}
Under Assumption~\ref{assum:sig} and the event-triggering condition \eqref{eq:etcondition}, there exists a minimum dwell-time $\tau_l>0$ such that for all $j\in\mathbb N$, the inter-event times satisfy
\begin{align}
t_{j+1}-t_j\ge \tau_l,
\label{tau-l-new}
\end{align}
where $\tau_l>0$ is independent of the initial condition.
\end{lem}
\begin{pf}
{ 
The proof is divided into three steps.

\emph{Step 1: triggering inequality at the next event.}
If $E_L(t_j)=\emptyset$, then no new event is generated after $t_j$ and there
is nothing to prove. Otherwise, let $t_{j+1}=\inf E_L(t_j)$. By the definition
of the triggering set and the continuity of the monitored quantities,
\begin{align}
\nu R_LV(t_{j+1})
\le
2\inner{g(\cdot)\mathcal{T}_j(w(t_{j+1},\cdot))}
{\widehat e_j(t_{j+1},\cdot)} .
\label{lemma5-1}
\end{align}

\emph{Step 2: growth bound for the monitored perturbation.}
First, the same Volterra estimate as in Proposition~\ref{lm:uniform_T}
applies uniformly to every frozen transformation $\mathcal T_j$, because
each frozen parameter vector belongs to the admissible set. Hence
\begin{align*}
\|\mathcal T_jw\|\le M_2\|w\|.
\end{align*}
By Proposition~\ref{lm:boundK} and \eqref{eq:kernel_approximation_bound},
the approximated kernels are uniformly bounded. We denote by 
\begin{align*}
    M_3 = \max_{j\in\mathbb{N}}\left\{\norm{\widehat{K}_j^{uu}}_{L^\infty(\mathcal D)}, \norm{\widehat{K}_j^{uv}}_{L^\infty(\mathcal D)}, \norm{\widehat{K}_j^{vu}}_{L^\infty(\mathcal D)}, \norm{\widehat{K}_j^{vv}}_{L^\infty(\mathcal D)}\right\} > 0,
\end{align*}
a constant depending only on this uniform kernel bound and on the fixed
Volterra integration domain. 

Assumption~\ref{assum:sig} gives, for $t\in[t_j,t_{j+1})$,
\begin{align*}
|\delta_{\sigma_j^\pm}(t,x)|
=|\sigma^\pm(t,x)-\sigma^\pm(t_j,x)|
\le l_\sigma(t-t_j),\qquad x\in[0,1].
\end{align*}
Using \eqref{eq:e1}, we obtain
\begin{align}
|\heone(t,x)|
&\le l_\sigma(t-t_j)|v(t,x)|
 +M_3l_\sigma(t-t_j)\int_0^x|v(t,\xi)|\,d\xi \nonumber\\
&\quad
 +M_3l_\sigma(t-t_j)\int_0^x|u(t,\xi)|\,d\xi .
\end{align}
For any $f\in L^2([0,1])$ and $x\in[0,1]$, Cauchy--Schwarz gives
\begin{align*}
\int_0^x|f(\xi)|\,d\xi
\le \left(\int_0^x1^2\,d\xi\right)^{1/2}
     \left(\int_0^x|f(\xi)|^2\,d\xi\right)^{1/2}
\le \|f\|.
\end{align*}
Taking the $L^2([0,1])$ norm in $x$ and absorbing the fixed numerical constants
from the two integral terms into $M_3$ gives
\begin{align}
\|\heone(t,\cdot)\|
&\le \widehat M_K\,l_\sigma(t-t_j)\|w(t,\cdot)\|.
\end{align}
where $\widehat M_K:=1+M_3$. The same argument applied to \eqref{eq:e2} gives the same bound for
$\hetwo$. Enlarging $M_3$, if necessary, by a fixed factor independent of
$j,t$, and the initial condition, the vector perturbation satisfies
\begin{align}
\|\widehat e_j(t,\cdot)\|
&\le \widehat M_K\,l_\sigma(t-t_j)\|w(t,\cdot)\|.
\end{align}
Combining this estimate with Cauchy--Schwarz gives
\begin{align}
2\inner{g(\cdot)\mathcal T_jw(t,\cdot)}{\widehat e_j(t,\cdot)}
&\le 2\|g\|_\infty
      \|\mathcal T_jw(t,\cdot)\|\,
      \|\widehat e_j(t,\cdot)\| \nonumber\\
&\le 2M_2\|g\|_\infty\widehat M_K
      l_\sigma(t-t_j)\|w(t,\cdot)\|^2 .
\end{align}
Using the lower norm equivalence in \eqref{eq:V_equiv},
\begin{align*}
\|w(t,\cdot)\|^2\le \frac{1}{m_V}V(t),
\end{align*}
we obtain
\begin{align}
2\inner{g(\cdot)\mathcal T_jw(t,\cdot)}{\widehat e_j(t,\cdot)}
\le C_1(t-t_j)V(t),
\label{lemma5-2}
\end{align}
where
\begin{align*}
C_1:=\frac{2M_2\|g\|_\infty\widehat M_Kl_\sigma}{m_V}.
\end{align*}

\emph{Step 3: dwell-time lower bound.}
If $l_\sigma=0$, then $\widehat e_j\equiv0$ and the triggering set is empty
unless the state is already at the zero solution. Thus no accumulation of
events can occur. For $l_\sigma>0$, evaluating \eqref{lemma5-2} at
$t=t_{j+1}$ and combining it with \eqref{lemma5-1} gives
\begin{align}
\nu R_LV(t_{j+1})
\le C_1(t_{j+1}-t_j)V(t_{j+1}).
\label{lemma5-3}
\end{align}
If $V(t_{j+1})=0$, then $w(t_{j+1},\cdot)=0$ by \eqref{eq:V_equiv}; the
closed-loop solution remains zero and no further event is needed. If
$V(t_{j+1})>0$, cancellation in \eqref{lemma5-3} yields
\begin{align}
t_{j+1}-t_j
\ge
\frac{\nu R_L}{C_1}
=: \tau_l>0.
\label{lemma5-4}
\end{align}
The lower bound depends only on the admissible constants and design
parameters, not on the initial condition or on $j$. Hence infinitely many
events cannot occur on a finite time interval. This finishes the proof of
Lemma~\ref{lem:minidwell}.
}
\end{pf}

\subsection{Stability analysis}
Next, we analyze the exponential stability of the closed-loop system consisting of the original system~\eqref{eq:sys1}-\eqref{eq:sys4} with the NO-approximated control law~\eqref{eq:nocontrolaw} under the event-triggering condition~\eqref{eq:etcondition}.

\begin{thm}\label{thm:stability}
    Under Assumption~\ref{assum:sig}, there exists $c_4>0$ and $C_r>0$, if the following condition is fulfilled,
    \begin{align}
        &c_4 < \frac{\e^{-\nu\phi_1(1)}}{a\rho^2\e^{\nu\phi_2(1)}}-1,\label{eq:c4condition}\\
        &l_\sigma < \frac{\nu R_L m_VC_\kappa + \nu R_L 2\norm{g}_\infty M_2 C_r \epsilon}{2M_2 \norm{g}_\infty \widehat{M}_K\ln{(\tfrac{M_V}{m_V})}}\label{eq:lcondition},
    \end{align}
    where $C_\kappa = \nu-\nu R_L - \tfrac{2}{m_V}(1+\tfrac{1}{c_4})a\e^{\nu\phi_2(1)}(1+\abs{\rho})^2 \epsilon^2$ and the approximation error $\epsilon$ is small enough, then for any initial condition $(u_0(\cdot),v_0(\cdot)) \in (L^2([0,1]))^2$, the system~\eqref{eq:sys1}-\eqref{eq:sys4} with the NO-approximated control law~\eqref{eq:nocontrolaw} under the event-triggering condition~\eqref{eq:etcondition} is exponentially stable in $L^2$-sense.
    Specifically, there exist constants $C_2 > 0$ and $\kappa_1 > 0$
    \begin{align}
        \norm{(u(t,\cdot),v(t,\cdot))}^2
        &\leq C_2 e^{-\kappa_1t} \norm{(u_0(\cdot),v_0(\cdot))}^2.
    \end{align}
\end{thm}
\begin{pf}
{ 
The proof is divided into four steps.

\emph{Step 1: Lyapunov derivative on a fixed inter-event interval.}
Fix $j$ and $t\in[t_j,t_{j+1})$. Define the boundary residual induced by the
NO kernel approximation as
\begin{align}
d_j(t):=\int_0^1 \big(\tKvu_j(1,\xi)-\rho\tKuu_j(1,\xi)\big)u(t,\xi)\,d\xi
+\int_0^1 \big(\tKvv_j(1,\xi)-\rho\tKuv_j(1,\xi)\big)v(t,\xi)\,d\xi .
\end{align}
Then \eqref{eq:tar4} can be written as
\begin{align*}
\beta_j(t,1)=\rho\alpha_j(t,1)+d_j(t).
\end{align*}
Write
\begin{align*}
V(t)=V_\alpha(t)+V_\beta(t),
\end{align*}
where
\begin{align*}
V_\alpha(t)=\int_0^1\frac{\e^{-\nu\phi_1(x)}}{\lambda(x)}
\alpha_j^2(t,x)\,dx,\qquad
V_\beta(t)=\int_0^1a\frac{\e^{\nu\phi_2(x)}}{\mu(x)}
\beta_j^2(t,x)\,dx .
\end{align*}
Using \eqref{eq:tar1}, namely
$\partial_t\alpha_j=-\lambda\partial_x\alpha_j+\heone+\rone$, we get
\begin{align}
\dot V_\alpha(t)
&=2\int_0^1\frac{\e^{-\nu\phi_1(x)}}{\lambda(x)}
       \alpha_j(t,x)\partial_t\alpha_j(t,x)\,dx \nonumber\\
&=-2\int_0^1\e^{-\nu\phi_1(x)}
       \alpha_j(t,x)\partial_x\alpha_j(t,x)\,dx \nonumber\\
&\quad
 +2\int_0^1\frac{\e^{-\nu\phi_1(x)}}{\lambda(x)}
       \alpha_j(t,x)(\heone(t,x)+\rone(t,x))\,dx .
\end{align}
Since $2\alpha_j\partial_x\alpha_j=\partial_x(\alpha_j^2)$ and
$\phi_1'(x)=1/\lambda(x)$, integration by parts gives
\begin{align}
\dot V_\alpha(t)
&=-\nu V_\alpha(t)
-\e^{-\nu\phi_1(1)}\alpha_j^2(t,1)
+\alpha_j^2(t,0) \nonumber\\
&\quad
+2\int_0^1\frac{\e^{-\nu\phi_1(x)}}{\lambda(x)}
       \alpha_j(t,x)(\heone(t,x)+\rone(t,x))\,dx .
\end{align}
Similarly, using \eqref{eq:tar2}, namely
$\partial_t\beta_j=\mu\partial_x\beta_j+\hetwo+\rtwo$, and
$\phi_2'(x)=1/\mu(x)$, one obtains
\begin{align}
\dot V_\beta(t)
&=-\nu V_\beta(t)
+a\e^{\nu\phi_2(1)}\beta_j^2(t,1)-a\beta_j^2(t,0) \nonumber\\
&\quad
+2a\int_0^1\frac{\e^{\nu\phi_2(x)}}{\mu(x)}
       \beta_j(t,x)(\hetwo(t,x)+\rtwo(t,x))\,dx .
\end{align}
Combining the two identities, using
$\alpha_j(t,0)=q\beta_j(t,0)$ and
$\beta_j(t,1)=\rho\alpha_j(t,1)+d_j(t)$, yields
\begin{align}
\dot V(t)
&= -\nu V(t)+(q^2-a)\beta_j^2(t,0)
+a\e^{\nu\phi_2(1)}\big(\rho\alpha_j(t,1)+d_j(t)\big)^2
-\e^{-\nu\phi_1(1)}\alpha_j^2(t,1) \nonumber\\
&\quad
+2\inner{g(\cdot)\mathcal T_j(w(t,\cdot))}
{\widehat e_j(t,\cdot)+r^\epsilon_j(t,\cdot)} .
\end{align}

\emph{Step 2: boundary residual estimate.}
For any $c_4>0$, Young's inequality gives
\begin{align*}
\big(\rho\alpha_j(t,1)+d_j(t)\big)^2
\le (1+c_4)\rho^2\alpha_j^2(t,1)
+\left(1+\frac1{c_4}\right)d_j^2(t).
\end{align*}
We now bound $d_j(t)$. From \eqref{eq:kernel_approximation_bound},
\begin{align*}
|\tKvu_j(1,\xi)-\rho\tKuu_j(1,\xi)|
\le (1+|\rho|)\epsilon,
\qquad
|\tKvv_j(1,\xi)-\rho\tKuv_j(1,\xi)|
\le (1+|\rho|)\epsilon .
\end{align*}
Therefore, by Cauchy--Schwarz on $[0,1]$,
\begin{align}
|d_j(t)|
&\le (1+|\rho|)\epsilon
     \int_0^1 |u(t,\xi)|\,d\xi
 +(1+|\rho|)\epsilon
     \int_0^1 |v(t,\xi)|\,d\xi \nonumber\\
&\le (1+|\rho|)\epsilon
     \big(\|u(t,\cdot)\|+\|v(t,\cdot)\|\big) \nonumber\\
&\le \sqrt2(1+|\rho|)\epsilon\|w(t,\cdot)\|.
\end{align}
Hence
\begin{align}
d_j^2(t)\le 2(1+|\rho|)^2\epsilon^2\|w(t,\cdot)\|^2
\le \frac{2(1+|\rho|)^2}{m_V}\epsilon^2V(t),
\end{align}
where \eqref{eq:V_equiv} was used in the last inequality. Substituting this
bound into the previous derivative identity gives
\begin{align}
\dot V(t)
&\le -\left(\nu-\frac{2}{m_V}\left(1+\frac1{c_4}\right)
a\e^{\nu\phi_2(1)}(1+|\rho|)^2\epsilon^2\right)V(t) \nonumber\\
&\quad +(q^2-a)\beta_j^2(t,0)
+\big(a\e^{\nu\phi_2(1)}(1+c_4)\rho^2-\e^{-\nu\phi_1(1)}\big)
\alpha_j^2(t,1) \nonumber\\
&\quad
+2\inner{g(\cdot)\mathcal T_j(w(t,\cdot))}{\widehat e_j(t,\cdot)}
+2\inner{g(\cdot)\mathcal T_j(w(t,\cdot))}{r^\epsilon_j(t,\cdot)} .
\label{dotLypaunov}
\end{align}

\emph{Step 3: absorption of boundary and in-domain perturbations.}
The parameter $a$ is chosen so that $a>q^2$. Condition~\eqref{eq:c4condition}
is equivalent to
\begin{align*}
a\e^{\nu\phi_2(1)}(1+c_4)\rho^2-\e^{-\nu\phi_1(1)}<0 .
\end{align*}
Thus the two boundary trace terms in \eqref{dotLypaunov} are nonpositive.
Next, the triggering condition gives, for all $t\in[t_j,t_{j+1})$,
\begin{align*}
2\inner{g(\cdot)\mathcal T_j(w(t,\cdot))}{\widehat e_j(t,\cdot)}
\le \nu R_LV(t).
\end{align*}
It remains to estimate $r_j^\epsilon$. By
\eqref{eq:assumption_bound_coupling_coeff},
\begin{align*}
|\delta_{\sigma_j^\pm}(t,x)|
\le |\sigma^\pm(t,x)|+|\sigma^\pm(t_j,x)|
\le 2M_\sigma .
\end{align*}
Using the definition of $\rone$ and the bound
$\|\widetilde K_j^\star\|_{L^\infty(\mathcal D)}\le\epsilon$, we obtain
\begin{align}
|\rone(t,x)|
&\le 2M_\sigma\epsilon\int_0^x |v(t,\xi)|\,d\xi
     +2M_\sigma\epsilon\int_0^x |u(t,\xi)|\,d\xi \nonumber\\
&\le 2M_\sigma\epsilon
     \big(\|u(t,\cdot)\|+\|v(t,\cdot)\|\big) \nonumber\\
&\le 2\sqrt2M_\sigma\epsilon\|w(t,\cdot)\|.
\end{align}
Taking the $L^2([0,1])$ norm in $x$ gives the same bound for
$\|\rone(t,\cdot)\|$. The same calculation applies to $\rtwo$, and hence
\begin{align*}
\|r_j^\epsilon(t,\cdot)\|
\le 4M_\sigma\epsilon\|w(t,\cdot)\|.
\end{align*}
Set $C_r:=4M_\sigma$. Then
\begin{align}
2\inner{g(\cdot)\mathcal T_j(w(t,\cdot))}{r_j^\epsilon(t,\cdot)}
&\le 2\|g\|_\infty
      \|\mathcal T_jw(t,\cdot)\|\,
      \|r_j^\epsilon(t,\cdot)\| \nonumber\\
&\le 2\|g\|_\infty M_2 C_r\epsilon\|w(t,\cdot)\|^2 \nonumber\\
&\le \frac{2\|g\|_\infty M_2 C_r\epsilon}{m_V}V(t).
\end{align}
Let
\begin{align*}
C_3:=\frac{2\|g\|_\infty M_2 C_r\epsilon}{m_V}.
\end{align*}
Using the definition of $C_\kappa$ in the theorem statement, the preceding
estimates imply
\begin{align*}
\dot V(t)\le -(C_\kappa+C_3)V(t),\qquad t\in[t_j,t_{j+1}).
\end{align*}
The smallness of the approximation error $\epsilon$ is used here to ensure
\begin{align*}
\kappa:=C_\kappa+C_3>0.
\end{align*}
Therefore,
\begin{align}
V_j(t)\le \e^{-\kappa(t-t_j)}V_j(t_j),
\qquad t\in[t_j,t_{j+1}).
\label{eq:interval_decay}
\end{align}
where $V_j$ emphasizes that the Lyapunov functional is built with the kernels
frozen at $t_j$.

\emph{Step 4: iteration across triggering times.}
The physical state is continuous at triggering instants because no reset is
applied to the plant. The Lyapunov functional may change when the kernels are
updated, but the norm equivalence \eqref{eq:V_equiv} gives
\begin{align*}
V_{j+1}(t_{j+1})
\le M_V\|w(t_{j+1},\cdot)\|^2
\le \frac{M_V}{m_V}V_j(t_{j+1}^{-}).
\end{align*}
Set $R_V:=M_V/m_V$. Applying \eqref{eq:interval_decay} on each inter-event
interval and multiplying the resulting estimates gives
\begin{align*}
V_j(t_j)\le R_V^j\e^{-\kappa t_j}V_0(0).
\end{align*}
For $t\in[t_j,t_{j+1})$, another use of \eqref{eq:interval_decay} yields
\begin{align*}
V_j(t)\le R_V^j\e^{-\kappa t}V_0(0).
\end{align*}
Using $V_0(0)\le M_V\|w(0,\cdot)\|^2$ and
$\|w(t,\cdot)\|^2\le V_j(t)/m_V$, we get
\begin{align*}
\|w(t,\cdot)\|^2
\le R_V^{j+1}\e^{-\kappa t}\|w(0,\cdot)\|^2 .
\end{align*}
By Lemma~\ref{lem:minidwell}, $j\le t/\tau_l$. Consequently,
\begin{align}
\|w(t,\cdot)\|^2
\le \frac{M_V}{m_V}
\exp\!\left[-\left(\kappa-\frac1{\tau_l}\ln\frac{M_V}{m_V}\right)t\right]
\|w(0,\cdot)\|^2 .
\end{align}
The condition~\eqref{eq:lcondition}, together with the smallness requirement
on $\epsilon$, ensures
\begin{align*}
\kappa_1:=\kappa-\frac1{\tau_l}\ln\frac{M_V}{m_V}>0 .
\end{align*}
Taking $C_2:=M_V/m_V$ gives the claimed exponential estimate. 
\begin{align}
        \norm{(u(t,\cdot),v(t,\cdot))}^2
        &\leq C_2 e^{-\kappa_1t} \norm{(u_0(\cdot),v_0(\cdot))}^2.
    \end{align}
This completes the proof of Theorem~\ref{thm:stability}.
}
\end{pf}
\begin{rem}
    The decay rate $\kappa_1$ of the Lyapunov-like functional $V(t)$ is explicitly determined by the approximation error $\epsilon$. While the system attains its maximum convergence rate when $\epsilon=0$, such an error-free case is unrealistic for neural operators in practice. This highlights a trade-off between the convergence speed $\kappa_1$ and the approximation error $\epsilon$, achieving a smaller $\epsilon$ to accelerate convergence requires careful parameter tuning and rigorous training of the neural operator.
\end{rem}

\begin{rem}
    The event-triggering condition~\eqref{eq:etcondition} is purely Lyapunov-based and relies on the in-domain perturbation $e_j(t,x)$, which depends on the time-variation of the couplings $\sgp(t,x),\sgn(t,x)$. In addition, when the backstepping kernels are generated by a neural operator, the boundary actuation involves an approximation error $\epsilon$ and the subsequent Lyapunov estimate. To explicitly account for this effect, we could use another triggering mechanism that augments the threshold by a term proportional to $\epsilon^2$. For each $j\in\mathbb{N}$,
    \begin{align}\label{eq:etcondition_error}
    E_\epsilon(t_j)
    := \big\{t>t_j:
    2\inner{g(\cdot)\mathcal{T}_j(w(t,\cdot))}{\widehat e_j(t,\cdot)}>
    \nu R_L V(t)+R_\epsilon\,\epsilon^2 V(t)\big\}.
\end{align}
    where  $R_\epsilon>0$ is a design parameter. The triggering times $\{t_j\}_{j\in\mathbb{N}}$ are then generated by the following rules:
    \begin{enumerate}
    \item If $E_\epsilon(t_j)=\emptyset$, then the sequence stops and the set of triggering times is $\{t_0,t_1,\dots,t_j\}$.
    \item If $E_\epsilon(t_j)\neq\emptyset$, then the next triggering time is defined as
    $t_{j+1}:=\inf E_\epsilon(t_j)$.
    \end{enumerate}
    The Zeno behavior can be also proved to be excluded through the same methods with a large minimal dwell-time $\tau_\epsilon := \tfrac{\nu R_L+ R_\epsilon\epsilon^2}{ C_1}$,
    and the stability of the closed-loop system is obtained following the same manner in the proof of Theorem~\ref{thm:stability} with proper selection of $l_\sigma$.
\end{rem}

\begin{rem}\label{rem:tradeoff_eps}
Compared with~\eqref{eq:etcondition}, the term $R_\epsilon\epsilon^2V(t)$ in condition~\eqref{eq:etcondition_error}
leads to a larger dwell-time lower bound but also results in a decay rate that explicitly depends on
$\epsilon$. In practice, $R_\epsilon$ can be tuned to balance triggering frequency and robustness margin with respect to the neural operator approximation error.
\end{rem}

\section{Simulation}\label{sec:Sim}
Aligned with  Section \ref{section:preliminaries_NO}, we conduct first the training process based on the nominal system \eqref{eq:nomsys1}-\eqref{eq:nomsys4}. Recall that we use neural operator to approximate the operator mapping from \{$\lambda(x)$, $\mu(x)$, $\bar\sigma^+(x)$, $\bar\sigma^-(x)$, $q$, $\rho$\} to backstepping kernels $K^{\cdot\cdot}(x,\xi)$. More precisely, the neural operator is trained using a dataset of 2000 samples generated by solving the backstepping kernel equations numerically for various realizations of $\bar\sigma^{i,+}(x)$ and $\bar\sigma^{i,-}(x)$, $i=1,2,\dots,2000$. The finite difference method was adopted for solving backstepping kernels numerically. The dataset consists of random sampling of $\bar\sigma^{+}(x)$, $\bar\sigma^{-}(x)$ during the simulation period. Then the neural operator takes the input-output pair for training. The proposed neural operator, implemented via DeepXDE with a PyTorch backend, consists of a branch network mapping input functions and a trunk network processing spatial coordinates into a shared latent dimension $p$. Both networks utilize MLP architectures featuring three hidden layers with a constant width of 256 units each with ReLU activation. The final operator output is reconstructed by computing the inner product between the output of branch net $[b_1,b_2,\dots,b_p]$ and trunk net $[z_1,z_2,\dots, z_p]$ to generate scalar values across the coordinate domain. The structure is shown in Fig.~\ref{fig:nostruc}. $(\Phi^i(x_1),\Phi^i(x_2),\dots, \Phi^i(x_m))^\top$ is the input function of branch network corresponding to the function $(\mathbf{u}(\mathbf{x}_1), \mathbf{u}(\mathbf{x}_2),\dots, \mathbf{u}(\mathbf{x}_m))^\top$ in Proposition~\ref{Deeptheo}. $m$ is the data points of the finite representation of the input function $\Phi(x)$. $(x,\xi)$ is the input of trunk network which corresponding to the $\mathbf{y}$ in Proposition~\ref{Deeptheo}. The neural operator is trained using the Adam optimizer with a learning rate of 0.001 for 1000 epochs. The training process is performed on a Python running an Intel core i9-12900K central processing unit (CPU) with a clock rate of 3.60 GHz, and graphics processing unit (GPU) device, NVIDIA GeForce RTX 4090.

Then, we conduct simulations for system~\eqref{eq:sys1}-\eqref{eq:sys3},\eqref{eq:sys4del_withNO} with the following parameters: $\lambda(x) = 1$, $\mu(x) = 2$, $q = 1.2$, $\rho=0.3$, $\sigma^+(t,x)= 2 + 0.3x\e^{2\sin(t)}$, $\sigma^-(t,x)= 1 + \tfrac{2}{\cosh^2{(t-7)^2}} + 0.6\cos{(\pi t)} + \tfrac{2}{\cosh^2{(5x)}}$. The initial condition is chosen as $u_0(x) = 0.1 \sin(\frac{3 \pi x}{L}), v_0(x) = -0.1\sin( \frac{3\pi x}{L})$. We implement the NO-based ETGS method with  boundary control  \eqref{eq:nocontrolaw}, and under the event-triggering condition~\eqref{eq:etcondition} to update the kernels at discrete triggering times. The event-triggering parameters are set as $R_L = 0.1$.  We also use the modified triggering condition \eqref{eq:etcondition_error} (with $R_\epsilon = 0.1$) and compare with the results of the Lyapunov-based ETGS method proposed in~\cite{espitia_event-triggered_2025}. The total simulation time is $T=15$s.
\begin{table}[!tbp]
    \centering
    \caption{The event-trigger performance comparison}
    \begin{tabular}{c c c c c}
    \hline\hline
       Method  & Events & Min & Mean & ACT(ms) \\
       \hline \hline
       Methods in \cite{espitia_event-triggered_2025} & 44 & 0.04 & 0.34  & 208\\
       NO with~\eqref{eq:etcondition}  & 37 & 0.05 & 0.40  & 1.1\\
       NO with~\eqref{eq:etcondition_error}  & 25 & 0.07 & 0.60  & 1.1\\
       \hline \hline
    \end{tabular}
    \label{tab:comparison_ETC}
\end{table}

Table~\ref{tab:comparison_ETC} summarizes the performance comparison between the three event-triggered gain-scheduling methods. The results show that the NO-approximated kernels significantly reduce the average computation time(ACT) with 189$\times$ speedup compared with the numerical solving method, making the proposed method more efficient for real-time applications. Additionally, the NO-based ETGS with the triggering condition~\eqref{eq:etcondition_error} achieves fewer triggering times compared to the method with~\eqref{eq:etcondition}, demonstrating the effectiveness of incorporating the approximation error into the triggering mechanism.
\begin{figure}
    \centering
    \includegraphics[width=0.8\linewidth]{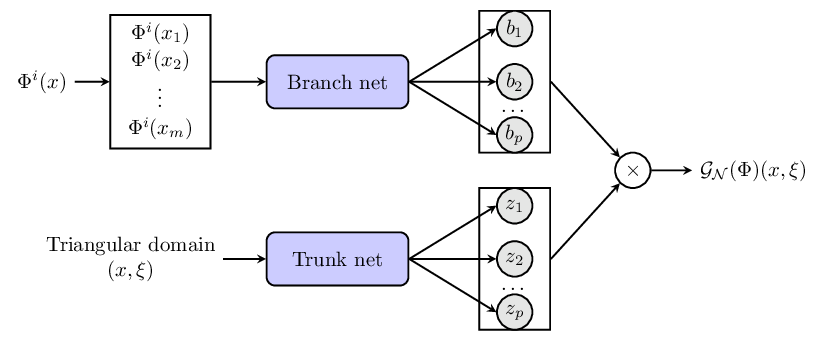}
    \caption{The structure of the neural operator}
    \label{fig:nostruc}
\end{figure}
\begin{figure}
    \centering
    \subfloat[State norm]{\includegraphics[width=0.4\linewidth]{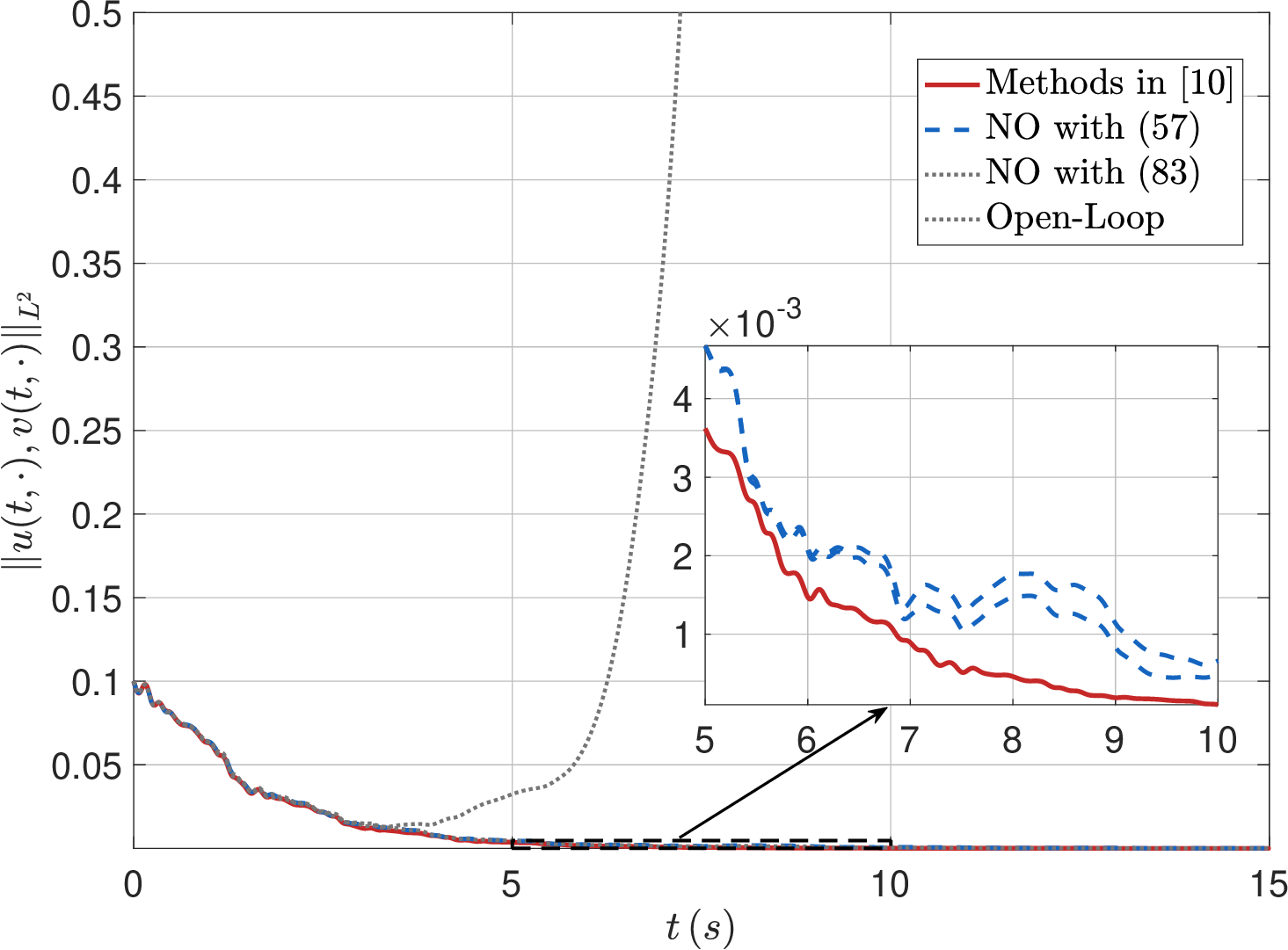}}
    \subfloat[Control input]{\includegraphics[width=0.42\linewidth]{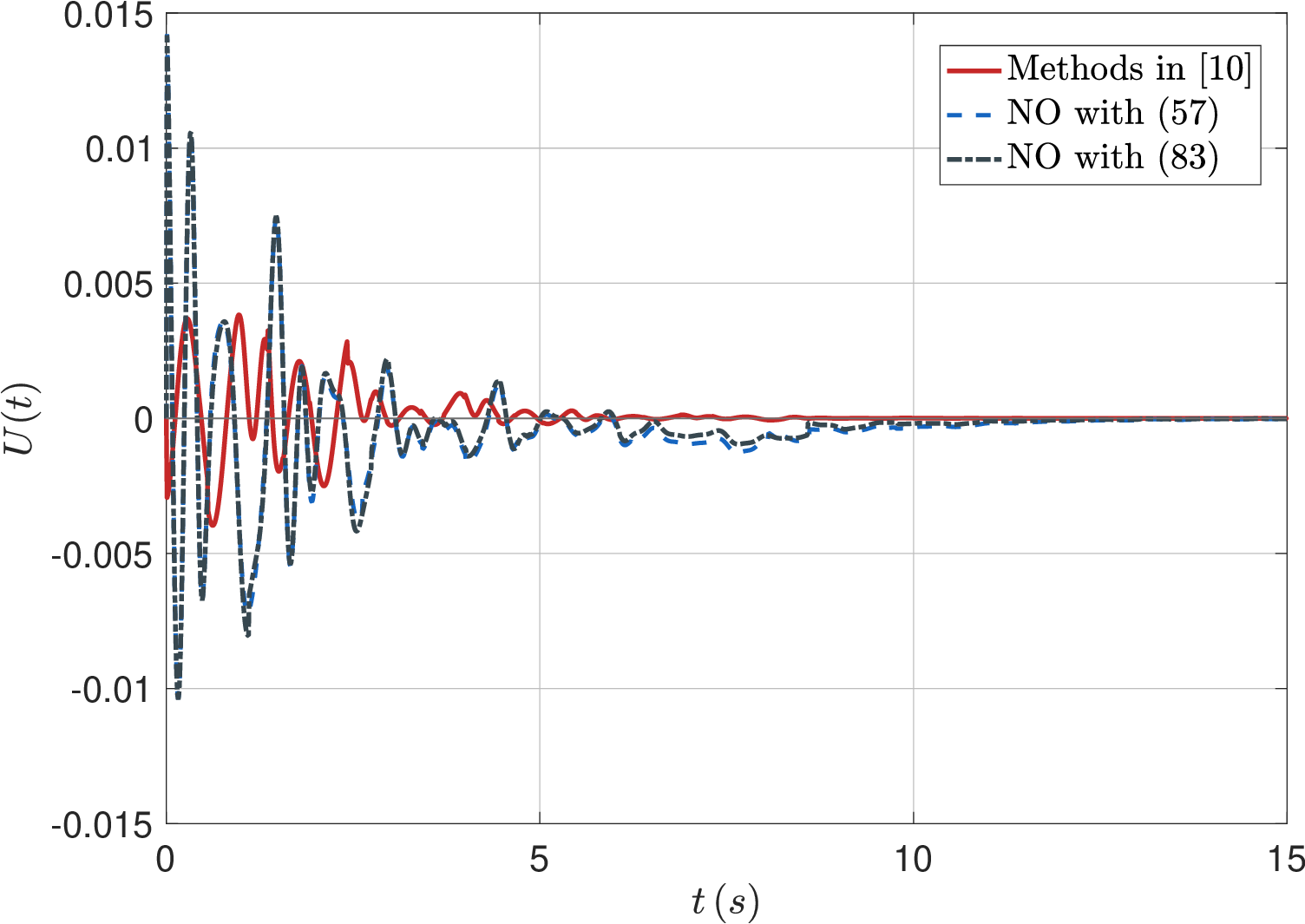}}
    \caption{The state norm evolution and control input evolution comparison between NO-based ETGS method and the gain-scheduling method in~\cite{espitia_event-triggered_2025} that updates the kernels at triggering time step.}
    \label{compareStateNorm}
\end{figure}
Figure~\ref{compareStateNorm} compares the state norm evolution and control input evolution between the NO-based ETGS and Lyapunov-based ETGS that updates the kernels at triggering time step. The results indicate that both methods achieve similar stabilization performance, with the NO-based ETGS method effectively maintaining system stability while reducing computational effort through fewer kernel updates. The control input comparison shows that both methods generate similar control inputs, demonstrating that the NO-based ETGS method effectively approximates the control law while reducing computational burden through fewer kernel updates.
\begin{figure}
    \centering
    \includegraphics[width=0.8\linewidth]{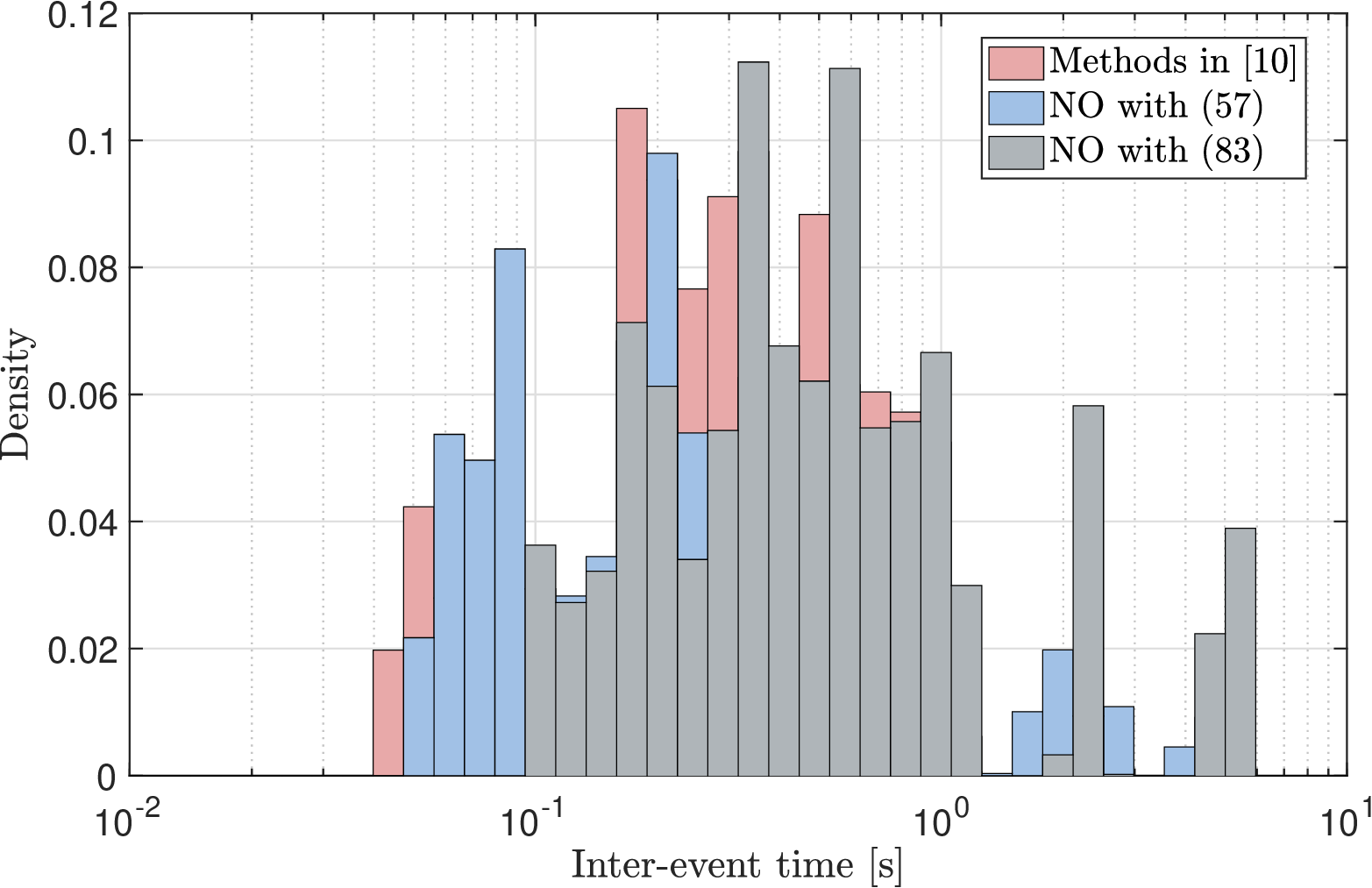}
    \caption{The triggering times and inter-event intervals under different event-triggering conditions.}
    \label{triggertime_comparison}
\end{figure}
{Figure~\ref{triggertime_comparison} shows the density of inter-execution times with 300 different initial conditions under three ETGS methods. It can be seen that the NO-based ETGS, with the triggering condition~\eqref{eq:etcondition_error}, leads to longer inter-execution times.}

\section{Conclusion}\label{sec:Con}
In this paper, we proposed a novel event-triggered gain scheduling method for $2 \times 2$ linear hyperbolic PDEs with time- and space-varying coefficients. The method leverages neural operators to approximate the backstepping kernels, enabling efficient real-time implementation of the control law. Simulation results demonstrated the effectiveness of the proposed method in stabilizing the system while significantly reducing computational effort compared to traditional numerical methods. Future work includes extending the proposed approach to more general classes of PDEs, and investigating the robustness of the method to model uncertainties and disturbances.

\bibliographystyle{abbrv}
\bibliography{reference}

\end{document}